\documentclass [12pt]{article}

\usepackage{fullpage,wrapfig}
\usepackage{amsmath, amsthm, amssymb, algorithm, thmtools, algpseudocode, enumerate, caption, framed}
\usepackage{paralist}

\usepackage[usenames,dvipsnames,svgnames,table]{xcolor}
\definecolor{darkgreen}{rgb}{0.0,0,0.9}

\usepackage[colorlinks=true,pdfpagemode=UseNone,citecolor=Blue,linkcolor=Red,urlcolor=BrickRed,
pagebackref]{hyperref}

\usepackage[T1]{fontenc}
\usepackage[utf8]{inputenc}
\usepackage{authblk}

\usepackage{tikz}

\newcommand{\gs}[1]{\texttt{G}(#1)}

\numberwithin{equation}{section}

\newtheorem{theorem}{Theorem}[section]
\newtheorem{observation}{Observation}[section]

\newtheorem{lemma}{Lemma}[section]

\newtheorem{corollary}[theorem]{Corollary}

\date{}

\title{On $r$-Guarding Thin Orthogonal Polygons\footnote{The work is supported by Natural Sciences and Engineering Research Council of Canada (NSERC).}}

\author{Therese Biedl}
\author{Saeed Mehrabi}
\affil{Cheriton School of Computer Science

University of Waterloo, Waterloo, Canada.

\texttt{biedl@cs.uwaterloo.ca}, \texttt{smehrabi@uwaterloo.ca}}

\begin{document}

\maketitle

\begin{abstract}
Guarding a polygon with few guards is an old and well-studied problem
in computational geometry.  Here we consider the following variant:
We assume that the polygon is orthogonal and {\em thin} in some sense, and 
we consider a point $p$ to guard a point $q$ if and only if the minimum
axis-aligned rectangle spanned by $p$ and $q$ is inside the polygon.

A simple proof shows that this problem is \textsc{NP}-hard on orthogonal 
polygons with holes, even if the polygon is thin. If there are no holes, then a thin polygon becomes a {\em tree}
polygon in the sense that the so-called dual graph of the polygon is a tree.
It was known that finding the minimum set of $r$-guards is polynomial for tree polygons, but
the run-time was $\tilde{O}(n^{17})$.
We show here that with a different approach the running time
becomes linear, answering a question posed by Biedl et al. (SoCG 2011). 
Furthermore, the approach is much more general, allowing to specify
subsets of points to guard and guards to use, and it generalizes to
polygons with $h$ holes or thickness $K$, becoming fixed-parameter
tractable in $h+K$.
\end{abstract}

\section{Introduction}
\label{sec:introduction}
The art gallery problem is one of the oldest problems studied in computational geometry. In the standard art gallery, introduced by Klee in 1973~\cite{orourke1987}, the objective is to observe a simple polygon $P$ in the plane with the minimum number of point guards, where a point $p\in P$ is seen by a guard if the line segment connecting $p$ to the guard lies entirely inside the polygon. Chv\'{a}tal~\cite{chvatal1975} proved that $\lfloor n/3\rfloor$ point guards are always sufficient and sometimes necessary to guard a simple polygon with $n$ vertices. The art gallery problem is known to be \textsc{NP}-hard on arbitrary polygons~\cite{lee1986} and orthogonal polygons~\cite{dietmar1995}. Even severely restricting the shape of the polygon does not help: the problem remains \textsc{NP}-hard for simple monotone polygons~\cite{krohn2013} and for orthogonal tree polygons 
(defined precisely below) if guards must be at vertices~\cite{tomas2013}.
Further, the art gallery problem is APX-hard on simple polygons \cite{eidenbenz2001}, but some approximation algorithms have been developed \cite{ghosh2010,krohn2013}.

A number of other types of guards have been studied, especially for orthogonal polygons. See for example guarding with sliding cameras~\cite{katz2011,durocherMFCS2013}, guarding with rectangles~\cite{franzblau1984} or with orthogonally convex polygons~\cite{motwani1989}. Also, different types of visibility have been studied, especially for orthogonal polygons: guards could be only seeing along horizontal or vertical lines inside $P$, or along an orthogonal staircase path inside $P$ \cite{motwani1989}, or use $r$-visibility (defined below).

\paragraph{\bf Definitions and Model.} Let $P$ be an orthogonal polygon with $n$ vertices.    
The {\em pixelation} of $P$ (also called {\em dent diagram} \cite{CulbersonReckhow1989} and related
to a {\em rectangleomino} \cite{biedl2011})
is the partition of $P$ obtained by extending a horizontal
and a vertical ray inward at every reflex vertex, and expand it until it hits the
boundary.  Let $\Psi$ be the resulting
set of rectangles that we call \emph{pixels} (also called {\em
basic regions} \cite{worman2007}).  See Figure~\ref{fig:pixelsRectangles}
for an example.  Note that $|\Psi|$ could be quadratic in general.
We will sometimes interpret the pixelation as a planar graph, with one vertex
at every corner of a pixel and an edge for each side of a pixel.
Define the {\em dual graph} $D$ of a polygon $P$ to be the weak dual graph of
the pixelation of $P$, i.e., $D$ has a vertex for every pixel and two pixels
are adjacent in $D$ if and only if they have a common side.

An orthogonal polygon $P$ is called a \emph{thin polygon} if any pixel-corner lies on the boundary of $P$. It is called a \emph{tree polygon} if 
its dual graph is a tree.  One can easily see that a tree polygon is the same
as a thin polygon that has no holes (see also Lemma~\ref{lem:thinSimpleTree}).
For most of this paper, polygons are assumed to be thin polygons.

We say that point $g$ {\em $r$-guards} a point $p$ if the 
minimum axis-aligned rectangle $R(g,p)$ containing $g$ and $p$ is a 
subset of $P$.
The (standard) {\em rGuarding problem} hence consists of finding
a minimum set $S$ of points such that any point in $P$ is $r$-guarded
by a point in $S$.  However, our results work for a broader problem as
follows.  Let $U\subseteq P$ be the region that we wish to guard.
In particular, we could choose to guard 
only the vertices of $P$, or only the boundary, or only those parts of
the art gallery that truly need to be watched.
Let $\Gamma$ be the set of guards that are allowed to be used (in particular,
we could choose to use only vertices as guards).  In the standard problem,
$\Gamma$ is the set of all points in $P$.  Biedl et al.~\cite{biedl2011}
introduced {\em pixel-guards},
where one guard consists of all the points that belong to 
one pixel (see Figure~\ref{fig:pixelsRectangles}).
Our approach allows pixel-guards, so $\Gamma\subset P\cup \Psi$. 
Now the {\em $(U,\Gamma,P)$-rGuarding problem} consists of finding a
minimum set $S$ of guards in $\Gamma$ such that all of $U$ is
$r$-guarded by some guard in $S$ (or to report that no such set exists).

Restricting the region that needs to be guarded exacerbates
some degeneracy-issues for $r$-guarding.
Previous papers were silent about what happens if
rectangle $R(g,p)$ (in the definition of $r$-guarding) is a line
segment.  For example, in Figure~\ref{fig:pixelsRectangles}, does
$g$ guard $u_4$?   Does $u_1$ guard $u_4$?     
This issue can be avoided by
assuming that only the interior of pixels must be guarded
(as seems to have been done by Keil and Worman \cite{worman2007},
e.g.~their Lemma 1 is false for point $u_4$ located in
the pixel $\psi_{10}$ in Figure~\ref{fig:pixelsRectangles}, because $u_4$ sees $q\in P$
but not all points in $\psi_{10}$ do).  When the entire polygon needs to be
guarded, then this is a reasonable restriction since the guards that see
the interior also see the boundary in the limit.
But if only a subset of $P$ must be guarded, then we must clarify how
degeneracies are to be handled.  We say that an axis-aligned
rectangle $R$ is {\em degenerate} if it has area 0 (i.e., is a line
segment) and there exists no rectangle $R'$ with positive area
and $R\subset R' \subseteq P$.  In Figure~\ref{fig:pixelsRectangles},
$R(g,u_4)$ is degenerate while $R(u_1,u_4)$ is not.  
Our approach is broad enough that it can handle
both allowing and disallowing the use of degenerate rectangles 
when defining $r$-guarding.

\begin{figure}[t]
\centering
\includegraphics[page=1,width=0.5\textwidth]{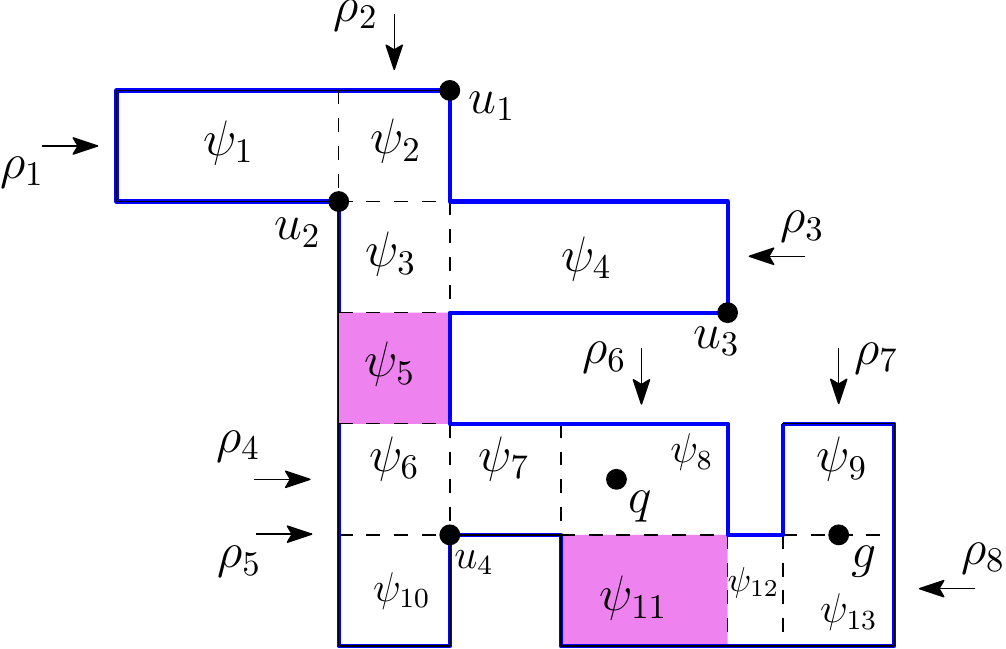}
\caption{A tree polygon with pixels $\{\psi_1,\dots,\psi_{13}\}$ and maximal axis-aligned rectangles $\{\rho_1,\dots,\rho_8\}$; rectangle $\rho_5$ is degenerate.
Pixel-guard $\psi_{5}$ guards $u_3$ via its top-right corner.
}
\label{fig:pixelsRectangles}
\end{figure}
\paragraph{\bf Related Results.}
The problem of guarding orthogonal polygons using $r$-guards was introduced by
Keil~\cite{keil1986} in 1986.
He gave an $O(n^2)$-time exact algorithm for the rGuarding problem for horizontally convex orthogonal polygons. The complexity of 
rGuarding in simple polygons
was a long-standing open problem until 2007 when Worman and Keil~\cite{worman2007} gave a polynomial-time algorithm for it. 
However, the algorithm by Worman and Keil is quite slow: it runs in $\tilde{O}(n^{17})$-time, where $n$ denotes the number of the vertices of $P$ and $\tilde{O}$ hides a poly-logarithmic factor. As such, Lingas et al.~\cite{Lingas2012} gave a linear-time 3-approximation algorithm for 
rGuarding in simple polygons.
Faster exact algorithms are known for a number of special cases of orthogonal polygons
\cite{keil1986,CulbersonReckhow1989,PaliosTzimas2014}.  All these algorithms require the polygon to be simple.
We are not aware of any results concerning the rGuarding problem for polygons with holes,
or if only the vertices or only the boundary need to be guarded or used as guards.

The first results on guarding thin polygons were (to our knowledge)
in \cite{biedl2011}; they studied guarding pixelations and asked
whether this can be done more easily if the dual graph is a tree.  However,
no better results than applying \cite{worman2007} were found.  
Later,
Tomas \cite{tomas2013} showed that indeed guarding tree polygons%
\footnote{Tomas constructs only simple polygons and hence used the term ``thin polygon'' for tree polygons.} 
is \textsc{NP}-hard in the traditional guarding-model (i.e. $g$ guards $p$ if the
line segment $gp$ is in $P$), and if all guards must be at vertices.
The complexity of guarding thin polygons in the $r$-guarding model remained
open.  Paper \cite{biedl2011} was also (apparently) the first paper to
consider pixel-guards in place of point-guards.

\paragraph{\bf Our Results.} 
In this paper, we resolve the complexity of the rGuarding problem on
thin polygons. We show with a simple reduction from Vertex Cover in
planar graphs that this problem is \textsc{NP}-hard on polygons with holes, even if the polygon is thin.
As our main result, we show that the rGuarding problem is linear-time solvable on thin
polygons without holes.

Comparing our results to the one by Worman and Keil~\cite{worman2007}, their algorithm works for a
broader class of polygons (they do not require thinness), but is slower.
Moreover, their approach crucially needs that the polygon is simple, 
that the entire polygon needs to be guarded, and that any point in the
polygon can guard.  In contrast to this, our approach generalizes easily
to a number of other scenarios.  First of all, it is not crucial that
the polygon is simple; we can deal with any constant number $h$ of holes.
Secondly, we can choose what to be guarded and what to guard
with; we can hence also solve all art gallery variants where only the
vertices or only the boundary need to be guarded, or where only guards
at the vertices or the boundary are allowed to be used.  Finally, the
restriction on thinness can be relaxed.  We use thinness only to bound
the treewidth of the dual graph of the polygon, and as long as the
treewidth is bounded the approach works.  In particular, if the polygon
is $K$-thin in some sense, and has at most $h$ holes, then for constants
$h$ and $K$ our algorithm is still linear, and the rGuarding problem
hence is fixed-parameter tractable in $h+K$.

\section{NP-hardness}
\label{apx:np}
In this section, we prove that rGuarding is \textsc{NP}-hard in polygons
with holes.
The reduction is from Vertex Cover in planar graphs with maximum
degree 3; it is well-known that this is \textsc{NP}-hard~\cite{GareyJ1977}.
So let $G=(V,E)$ be a planar graph with maximum degree 3.  
Let $G^s$ be the graph obtained from $G$ by subdividing every
edge twice.  It is folklore (see e.g.~\cite{Poljak74}) that 
$G$ has a vertex cover of size $k$ if and only if $G^s$ has 
a vertex cover of size $|E|+k$.  
$G$ has a planar orthogonal drawing with at most one bend per edge
(see e.g.~\cite{Kant96}).
By placing one subdivision vertex of each edge at such a bend (if any)
and placing the other subdivision vertex arbitrarily, we hence obtain
a drawing $\Gamma$ of $G^s$ where every vertex
is a point, every edge is a horizontal or vertical line segment, and 
edges are disjoint except at common endpoints.

We construct a polygon $P$ as a ``thickened'' version
of $\Gamma$.
%
After possible scaling, we may assume that $\Gamma$ resides in an
integer grid with consecutive grid-lines at least $2n$ units apart,
where $n=|V|$.
Replace each horizontal edge $e$ by a rectangle $R_e$ of unit height, 
spanning between the points corresponding to the ends of $e$.
Similarly replace each vertical edge by a rectangle of unit width.  These
rectangles will get moved later, but never so far that they
would overlap edge-rectangles from other rows or columns.

We replace vertex-points by small gadgets as illustrated in
Figure~\ref{fig:NPhard}.
Thus, let $v$ be a vertex of degree 3 in $G^s$; up to rotation it
has incident edges $e_1,e_2,e_3$ on the left, right and top in $\Gamma$.
Replace $v$ by two pixels, attach $R_{e_3}$ at the top of the upper pixel,
$R_{e_1}$ at the left side of the upper pixel and $R_{e_2}$ at the right
side of the lower pixel.  Let $s_v$ be the side common to the two pixels
of $R_v$.    
Rectangles $R_{e_1}$ and $R_{e_2}$ are not
quite horizontally aligned, resulting in one of them being offset from
the grid-line.  However, in total over all vertices in the row, there are at most
$n$ offsets, and so edge-rectangles remain disjoint.
For any vertex of degree 2, omit the
third rectangle  
and also any pixel that is not needed.  

\begin{figure}[ht]
\hspace*{\fill}
\includegraphics[width=90mm,trim=0 100 0 0,clip]{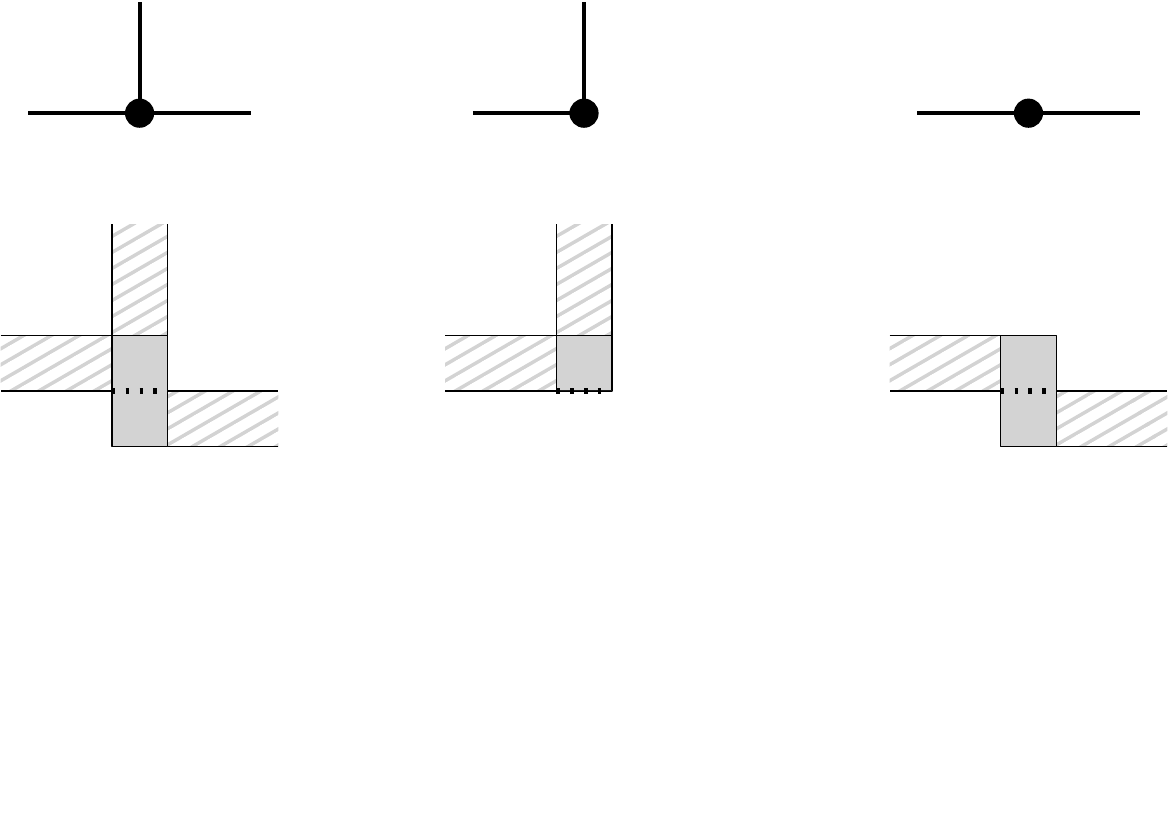}
\hspace*{\fill}
\caption{Converting an orthogonal drawing without bends into a polygon
for rGuarding. 
$R_e$ is hatched, $R_v$ is gray, and $s_v$ is dotted.}
\label{fig:NPhard}
\end{figure}

\begin{observation}
For any vertex $v$, any point in $s_v$ guards the rectangles $R_e$ of any
incident edge $e=(v,w)$, as well as the pixel of $w$ where $R_e$ attaches.

For any edge $e=(v,w)$, if any point in $R_e$ is $r$-guarded from a point
$q$, then $q$ belongs to $R_e$, $R_v$ or $R_w$.
\end{observation}

Using this observation, the reduction is immediate.  Namely, let $C$ be
a vertex cover of $G_s$ of size $k$.  For any $v\in C$, place a guard
anywhere along $s_v$.  Since $C$ was a vertex cover, this $r$-guards $R_e$
for all edges, and also $R_w$ for all $w\not\in C$ since each pixel of
$R_w$ is attached to some $R_e$.  Vice versa, if we have a 
set $S$ of $r$-guards, then we can create a set $C$ as follows:  For any vertex $v$, if
$R_v$ contains a guard in $S$, then add $v$ to $C$.  For any edge $e=(v,w)$,
if $R_e$ contains a guard in $S$ that is in neither $R_v$ nor $R_w$, then
arbitrarily add one of $v,w$ to $C$.    Clearly $|C|\leq |S|$, and since
any rectangle $R_e$ was guarded, any edge in $E$ is covered by $C$.

Inspection of Figure~\ref{fig:NPhard} shows that the constructed
polygon is thin.  Observe that it has holes, namely,  one per face of $G$. 
Since rGuarding is clearly in \textsc{NP}, we can conclude:

\begin{theorem}
The rGuarding problem is \textsc{NP}-complete on thin polygons.
\end{theorem}

\section{Polygons Whose Dual Has Bounded Treewidth }
\label{sec:treedecomp}

We now show how to solve the rGuarding problem in a tree polygon in
linear time.  In fact, we show something stronger, and prove that the
rGuarding problem can be solved in linear time in any polygon
for which the dual graph $D$ has bounded treewidth, and under any
restriction on the set $U$ to be guarded and the set $\Gamma$ that may
serve as guards.

The approach is to construct an auxiliary graph 
$H$, and argue that solving the rGuarding problem reduces to
a graph problem in $H$. Then we argue that the treewidth of $H$ satisfies $tw(H)\in O(tw(D))$ 
and that the graph problem is linear-time solvable in bounded treewidth graphs.
This auxiliary graph is different from the so-called region-visibility-graph
used by Worman and Keil~\cite{worman2007} in that it encodes who can guard
what, rather than who can be guarded by a common guard.

\subsection{Simplifying $U$ and $\Gamma$}

We first show that we can simplify
the points to guard and the point-guards to use such that only
a constant number of each occur at each pixel.  

\begin{lemma}
\label{lem:simplify}
Let $U\subseteq P$ be any (possibly infinite) set of points in $P$.
Then there exists a finite set of points $U'\subseteq U$ such that
$U'$ is $r$-guarded by a set $S$ if and only if $U$ is.  Moreover,
for any pixel $\psi$, at most 4 points in $U'$ belong to $\psi$.
\end{lemma}
\begin{proof}
We construct the set $U'$ as follows. 
\begin{itemize}
\item For every pixel $\psi$, if the interior of $\psi$ intersects $U$, then add one point from this intersection into $U'$. 
\item For every pixel-side $e$, if neither incident pixel has a point of $U$ in its interior, but the open set $e$ intersects $U$, 
	then add one point from this intersection to $U'$,
\item For every pixel-corner $c$, if $c\in U$, and if none of the incident pixels or pixel-sides has added a point to $U'$,
	then add $c$ to $U'$.
\end{itemize}
Thus, for any point $p$ in $U$, we have a point $p'$ to $U'$ such that
(i) if $p$ is in the interior of a pixel,
then so is $p'$, (ii) if $p$ is in the interior of some pixel-side $e$
then $p'$ is in the interior of $e$ or of some incident pixel,
and (iii) if $p$ is at a pixel-corner, then $p'$ is at that corner,
in the interior of an incident pixel-side, or in the interior of an incident pixel.
One can show (see Appendix~\ref{apx:preliminA}) that point $p'$ is
more restrictive with respect to guarding, i.e., any guard
that $r$-guards $p'$ also $r$-guards $p$.
\end{proof}

\begin{lemma}
\label{lem:simplifyGuards}
Let $\Gamma \subseteq P$ be any (possibly infinite) set of points in $P$.
Then there exists a finite set of points $\Gamma'\subseteq \Gamma$ 
such that for any pixel $\psi$, at most 4 points in $\Gamma'$ belong to $\psi$.
Moreover, if some set $S \subseteq \Gamma$ $r$-guards a set $U\subseteq P$, then
there exists a set $S'\subseteq \Gamma'$
with $|S'|\leq |S|$ that also $r$-guards $U$.
\end{lemma}
\begin{proof}
(Sketch) Similarly as the previous case, one argues that
we need at most one point-guard per interior, side or corner of each pixel.
A full proof can be found in Appendix~\ref{apx:preliminA}.
\end{proof}

\subsection{Maximal Rectangles and an Auxiliary Graph}  
Assume we are given a polygon $P$, a region $U\subseteq P$ to be
guarded, and a set $\Gamma$ of guards allowed to be used.
In what follows, we treat any element $\gamma\in \Gamma$ as a set,
so either $\gamma=\psi$ is a pixel-guard or $\gamma=\{p\}$ is
a point-guard.

As a first step, 
apply Lemmas~\ref{lem:simplify} and \ref{lem:simplifyGuards} to reduce
$U$ and the point-guards in $\Gamma$ so that they are finite sets,
each pixel contains at most 4 points of $U$,
and at most 4 point-guards of $\Gamma$.  

Let $\mathcal{R}$ be the set of maximal axis-aligned rectangles in $P$, i.e.,
$\rho \in \mathcal{R}$ if and only if $\rho\subseteq P$ and there is no
axis-aligned rectangle $\rho'$ with $\rho \subset \rho' \subseteq P$.  
In this definition of $\mathcal{R}$, we use the one that was meant for
$r$-guarding, i.e., we include degenerate rectangles in $\mathcal{R}$
if and only a degenerate rectangles $R(g,p)$ is sufficient for $g$
to $r$-guards $p$.

Now define graph $H$ as follows.  The
vertices of $H$ are $U\cup \mathcal{R} \cup \Gamma$, i.e.,
we have a vertex for every point that needs guarding, every maximal
rectangle in $P$, and every potential guard.  
We define edges of $H$ via containment as follows (see also Figure~\ref{fig:graphH}):
\begin{itemize}
\item There is an edge from a point $u\in U$ to a rectangle $\rho \in \mathcal{R}$ if and only
	if $u\in \rho$.
\item There is an edge from a potential guard $\gamma \in \Gamma$ to a rectangle $\rho\in \mathcal{R}$ if
	and only if their intersection is non-empty.
\end{itemize}
\begin{figure}[t]
\centering
\includegraphics[width=0.40\textwidth]{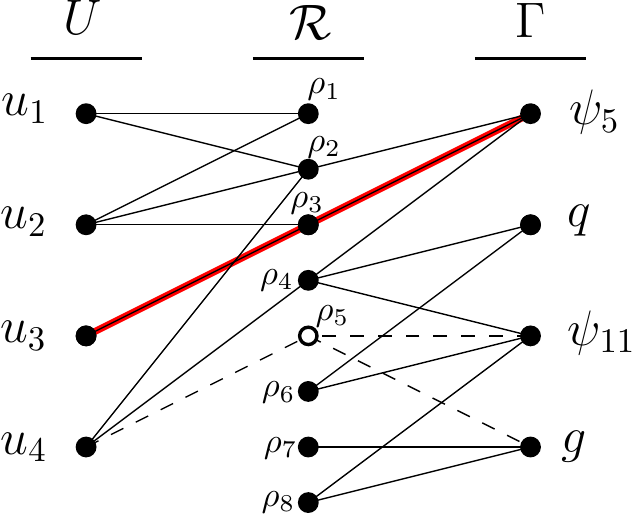}
\caption{The graph $H$ corresponding to Figure~\ref{fig:pixelsRectangles} for the chosen $U$ and $\Gamma$. The thick red path corresponds to
the pixel-guard $\psi_5$ seeing the point $u_3$ since both intersect rectangle $\rho_3$. Rectangle $\rho_5$ and its incident edges 
are included in $H$ only if we allow degenerate rectangles.}
\label{fig:graphH}
\end{figure}

\begin{lemma}
A point $u\in U$ is $r$-guarded by $\gamma\in \Gamma$ if and only if there exists a path of length 2 from $u$ to $\gamma$ in $H$.
\end{lemma}
\begin{proof}
If $u$ is $r$-guarded by $\gamma$, then there exists some $g\in \gamma$
such that the axis-aligned rectangle $R$ spanned by $p$ and $g$ is inside
$P$.  Expand $R$ until it is maximal to obtain $\rho\in\mathcal{R}$.
More precisely, if $R$ is non-degenerate, then use as $\rho$ some
maximal rectangle that has non-zero area and contains $R$.  If $R$ is
degenerate, then obviously degenerate rectangles were allowed for
$r$-guarding, and so expanding $R$ into a maximal line segment within $P$
gives an element $\rho$ of $\mathcal{R}$.
Either way $u\in R\subseteq \rho$ and $g\in R\subseteq \rho$ and we have a
path $u-\rho-g$ in $H$.

Vice versa, if there exists such a path, then it must have the form 
$u-\rho-\gamma$ for some maximal rectangle $\rho$ 
by construction of $H$.  By definition of the edges $u\in \rho$ and
some point $g\in \gamma$ satisfies $g\in \rho$, which means that the
axis-aligned rectangle spanned by $u$ and $g$ is inside $\rho\subseteq P$
and so $g$  (and with it $\gamma$) guards $u$.
\end{proof}

So the rGuarding problem reduces to finding the minimum
subset $S\subseteq \Gamma$ such that all $u\in U$ have a path of length~2
to some $\gamma\in S$, or reporting that no such $S$ exists.  
We call this the {\em restricted distance-$2$-dominating set} problem since it is the 
distance-$2$-dominating set while restricting who can
be chosen and who must be dominated.  We hence have:

\begin{lemma}
\label{lem:dist2dominating}
The $(U,\Gamma,P)$-rGuarding problem has a solution of size $k$
if and only if the restricted distance-$2$-dominating set in $H$ has
a solution of size $k$.
\end{lemma}

\subsection{Constructing a Tree Decomposition}
\label{sec:treedecompConstruction}

Assume now that the dual graph $D$ has small treewidth, defined as
follows. A
{\em tree decomposition} of a graph $D$ consists of a tree $I$ and
an assignment ${\cal X}:I\rightarrow 2^{V(D)}$ of {\em bags} to the nodes
of $I$ such that (a) for any vertex $v$ of $D$, the bags containing $v$
form a connected subtree of $I$ and (b) for any edge $(v,w)$ of $D$,
some bag contains both $v$ and $w$.  The width of such a decomposition
is $\max_{X\in {\cal X}} |X|-1$, and the {\em treewidth} $tw(D)$ of $D$ is the 
minimum width over all tree decompositions of $D$.

Fix a tree decomposition ${\cal T}=(I,{\cal X})$ of $D$ that has width $tw(D)$.
We now construct a tree decomposition of $H$ from $\mathcal{T}$
while increasing the bag-size by a constant factor.
Any bag $X\in {\cal X}$ consists of vertices of $D$, i.e., pixels
of $P$.
To obtain ${\cal T'}=(I,{\cal X'})$, modify any bag $X\in 
	{\cal X}$ to get $X'$ as follows:  For any pixel $\psi\in X$, add to $X'$
\begin{itemize}
\item any point of $U$ that is in $\psi$, 
\item any guard of $\Gamma$ that intersects $\psi$, and
\item any rectangle in $\mathcal{R}$ that intersects $\psi$.  
\end{itemize}

Finally we may (optionally) delete all pixels from all bags, since these are
not vertices of $H$.  We call the final construction ${\cal T}^H=(I,{\cal X}^H)$.
See also Figure~\ref{fig:treeDecomposition}.

\begin{figure}[t]
\centering
\includegraphics[page=2,width=0.34\textwidth,trim=120 0 70 0,clip]{treeDecomposition}
\hspace*{\fill}
\includegraphics[page=1,width=0.64\textwidth]{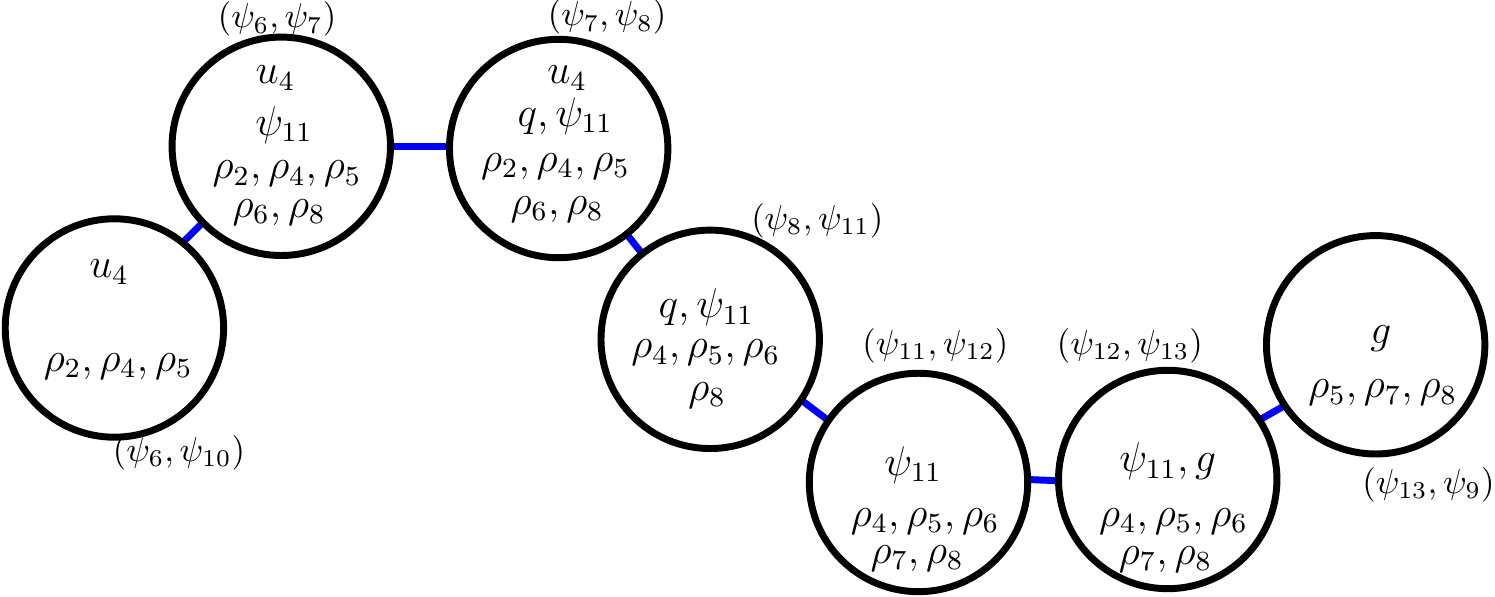}
\caption{The tree decomposition ${\cal T}^H=(I,{\cal X}^H)$ of graph $H$ corresponding to a sub-polygon of the one in Figure~\ref{fig:pixelsRectangles}.
We label the bags with the edges of the tree they correspond to.}
\label{fig:treeDecomposition}
\end{figure}

\begin{lemma}
For any polygon, ${\cal T}^H=(I,{\cal X}^H)$ is a tree decomposition of $H$.
If $P$ is thin, then the tree decomposition has width $O(tw(D))$.
\end{lemma}
\begin{proof}
First we argue that for any vertex of
$H$ the bags containing it are connected.  Crucial for this is that for
any pixel $\psi$, the bags that used to contain $\psi$ in ${\cal T}$ are
a connected subtree since ${\cal T}$ was a tree decomposition.
First consider a point $p$.  (We use $p$ for both
the point and for the vertex in $H$ representing it.)  Vertex $p$
was added to all bags that contained a pixel $\psi$ with $p\in \psi$.
There may be multiple such pixels (if $p$ is on the side or the
corner of a pixel), but the union of them
is a connected subgraph of $D$. 
For any connected subgraph, the bags containing vertices of it form
a connected subtree.  So the bags to which $p$ has been added form a connected subtree
of the tree $I$ of the tree decomposition as required.

The connectivity-argument is identical for a point-guard, and similar for pixel-guards and rectangles.  Namely,
consider a vertex of $H$ representing a pixel-guard $\gamma$.  This
guard was added to all the bags that contained a pixel $\psi$ that
intersects $\gamma$.  Again there may be many such pixels (up to 9),
but they are connected via $\psi$ and so the bags to which $\gamma$ is added
are connected.
Finally, consider a rectangle $\rho\in \mathcal{R}$ which was added to all
bags of pixels intersecting $\rho$.  The pixels that $\rho$
intersects form a connected subset of $P$ (because they are connected
along $\rho$), and hence correspond to a connected subgraph of $D$.
So the bags containing $\rho$ form a connected subtree.  

Now we must verify that for any edge of $H$, both endpoints appear in a bag.
Let $(u,\rho)$ be an edge from some point $u$ to some rectangle $\rho$.
Let $\psi$ be a pixel containing $u$.  Then $\rho\cap \psi \supseteq \{u\}$
is non-empty and so $\rho$ was added to any bag containing $\psi$.  We also
added $u$ to any bag containing $\psi$, so $u$ and $\rho$ appear in one bag.
Now consider some edge $(\gamma,\rho)$ from a
guard $\gamma$ to some rectangle $\rho$.
This edge exists because some point $g\in \gamma$ belongs to $\rho$.
Again fix some pixel $\psi$ that contains $g$ and observe that any
bag that contained $\psi$ has both $g$ and $\rho$ added to it.

It remains to discuss the width of the tree decomposition.  Consider
a bag $X$ of $\mathcal{T}$ and one pixel $\psi$ in $X$.
Since we reduced $U$ and $\Gamma$ with
Lemma~\ref{lem:simplify} and \ref{lem:simplifyGuards}, pixel $\psi$
intersects at most 4 points in $U$ and at most 4 point-guards.
It also intersects at most 9 pixel-guards.  Finally, one can show that in a thin
polygon $\psi$ intersects at most 6 maximal rectangles.
(A more general
statement will be proved in Lemma~\ref{lem:Kthin}.)  
Thus when creating bag $X'$ from bag $X$ we add
$O(1)$ new items per pixel and hence $|X'|\in O(|X|)$
and $\mathcal{T}^H$ has width $O(tw(D))$.
\end{proof}

\subsection{Solving 2-dominating Set}

To solve the restricted distance-$2$-dominating set problem on $H$, we first show that the problem can be expressed as a monadic second-order logic 
formula~\cite{Courcelle1990}. In particular, a set $S$ is a feasible solution for this problem if and only if
\[
S\subseteq \Gamma \quad\land\quad \forall u\in U \: \exists \rho\in \mathcal{R} \: \exists \gamma\in S: \: \text{adj}(u,\rho) \land \text{adj}(\rho, \gamma) 
\]
where adj is a logic formula to encode that its two parameters are adjacent in $H$. Since $H$ has bounded treewidth, we can find the smallest set $S$ that satisfies this or report that no such $S$ exists in linear time using Courcelle's theorem~\cite{Courcelle1990}. 
Here ``linear'' refers to the number of bags and hides a term that only depends on the treewidth.
One can show that a thin polygon has $O(n)$ pixels
(we will show something more general in Lemma~\ref{lem:kThinPixels}). 
Therefore graph $D$ has $O(n)$ vertices and hence a tree decomposition with $O(n)$
bags, and the run-time is hence $O(f(tw(D))n)$    for some computable function $f$.

\subsection{Run-time considerations}

We briefly discuss here how to do all other steps in linear time, under some
reasonable assumptions.  The first step is to find the pixels.  To do so,
we need to compute the vertical decomposition (i.e., the partition obtained
by extending only vertical rays from reflex vertices), which can be done in
$O(n)$ time \cite{Chazelle1991}.  Likewise, compute the horizontal decomposition.
Since (in a thin polygon) none of the rays intersect, we can obtain the pixels
(and with it, the pixelation-graph and $D$) in linear time.  Since
$D$ is planar, we can compute an $O(1)$-approximation of its treewidth in linear
time \cite{KammerTholey2012}, and hence can find $\mathcal{T}$ with width $O(tw(D))$.

Next we need to simplify $U$ and $\Gamma$.
The run-time to do so depends on the exact form of the original $U$ and $\Gamma$,
but as long as those have a simple enough form that we can answer queries such as
``does the interior of pixel $\psi$ intersect $U$'' in constant time, the overall
time is $O(1)$ per pixel and hence overall linear.

Next we need to find the rectangles $\mathcal{R}$.  In a thin polygon, all maximal
rectangles are either a ``slice'' defined by the vertical or horizontal decomposition,
or are a maximal line segment composed of pixel sides.  All such slices and
maximal line segments can be found from the pixelation in linear time, and there
are $O(n)$ of them.  This may yield some rectangles that are not maximal, but we
can retain those without harm since even then any pixel intersects $O(1)$ rectangles.

Constructing $H$ from these three sets, and building $\mathcal{T}^H$ given 
$\mathcal{T}$, can also clearly be done in linear time.
Putting everything together, we hence have:
\begin{theorem}
\label{thm:thin_treewidth}
Let $P$ be a thin polygon for which the dual graph has
treewidth $k$.  Then for any set $U\subseteq P$ and $\Gamma\subseteq P\cup \Psi$, we can solve the $(U,\Gamma,P)$-rGuarding problem in time
$O(f(k) n)$ time for some computable function $f$.
\end{theorem}

\section{Generalizations}
\label{sec:generalization}

In this section, we give some applications and generalizations of
Theorem~\ref{thm:thin_treewidth}.

\subsection{Thin Polygons with Few Holes} 

We claimed earlier that a simple thin polygon is a tree polygon,
and give here a formal proof because it will be useful later.

\begin{lemma}
\label{lem:thinSimpleTree}
Let $P$ be a thin polygon.  If $P$ has no holes, then the dual graph $D$
of the pixelation of $P$ is a tree.
\end{lemma}
\begin{proof}
Assume for contradiction that $D$ contains a cycle.  By tracing
along the midpoints of the pixels-sides corresponding to this cycle, 
we can create a simple closed curve 
$C$ that is inside $P$, yet has pixel-corners both inside and
outside $C$.  In a thin polygon, all pixel-corners are on the
boundary of $P$, so the boundary of $P$ has points both inside
and outside a simple closed curve that is strictly within $P$.  This
is possible only if $P$ has holes.
\end{proof}

Since every tree has treewidth 1, we hence have:

\begin{corollary}
Let $P$ be a thin polygon that has no holes.
Then for any sets $U\subseteq P$ and $\Gamma\subseteq P\cup \Psi$, we can solve 
the $(U,\Gamma,P)$-rGuarding problem in $O(n)$ time.
\end{corollary}

Inspecting the proof of Lemma~\ref{lem:thinSimpleTree}, we see that in fact every 
cycle of $D$ gives rise to a hole that is inside the curve defined by the
cycle.  If $D$ has $f$ inner faces, then each face defines a cycle in $D$,
and the insides of these cycles are disjoint.  Therefore, $D$ has at least $f$
holes.  Turning things around, if the polygon has $h$ holes, then $D$ has
at most $h$ inner faces.  In consequence, $D$ is a so-called 
$h$-outerplanar graph (i.e., if we remove all vertices from the outer-face and
repeat $h$ times, then all vertices have been removed).  It is well-known
that $h$-outerplanar graphs have treewidth $O(h)$ (see e.g.~\cite{cygan2015}).

\begin{corollary}
\label{cor:thin_treewidthFixedHoles}
Let $P$ be a thin polygon with $h$ holes.
Then for any sets $U\subseteq P$ and $\Gamma\subseteq P\cup \Psi$, we can solve the $(U,\Gamma,P)$-rGuarding problem in time
$O(f(h) n)$ time for some computable function $f$.
\end{corollary}

\subsection{Polygons That Are Not Thin}

The construction of the tree decomposition of $H$ in Section~\ref{sec:treedecompConstruction}
works even if $P$ is not thin.
However, the bound on the resulting treewidth, and the claim on the
linear run-time both used that the polygon is thin.  We can generalize these results to
polygons that are somewhat thicker.  More precisely, we say that a polygon is {\em $K$-thin}
(for some integer $K\geq 1$) if the dual graph $D$ of $P$ contains
no induced $(K+1)\times (K+1)$-grid.  A thin polygon is a $1$-thin polygon in this terminology,
because a pixel-corner is in the interior if and only if the four pixels around it form a
4-cycle, hence a $2\times 2$-grid, in $D$.
Notice that $K$-thin is equivalent to saying that the pixelation-graph
has no induced $(K+2)\times (K+2)$-grid. 
We need some observations:

\begin{lemma}
\label{lem:kThinDistance}
Let $P$ be a $K$-thin polygon.  Then, for any pixel-corner $p$, there exists
a point on the boundary of $P$ that is in the first quadrant relative to $p$
and has distance at most $2K+1$ from $p$,
where distance is measured by the length of the path in the pixelation-graph.
\end{lemma}
\begin{proof}
Consider any path in the pixelation graph that starts at $p$ and goes upward
or rightward for at most $K+1$ edges each.
If some such path reaches a point on the boundary after at most $2K+1$ edges,
then we are done.  Else the union of these paths forms a
$(K+2)\times (K+2)$-grid in the pixelation-graph, and $P$ is not $K$-thin.
\end{proof}

\begin{lemma}
\label{lem:kThinPixels}
The pixelation of a $K$-thin polygon with $n$ vertices has $O(K^2 n)$ pixels.
\end{lemma}
\begin{proof}
There are $O(n)$ boundary vertices: one for each vertex of $P$, and one
whenever a ray hit the boundary (of which there are at most $n-4$ since
there are $n/2-2$ reflex vertices and each emits two rays).
Each vertex on the boundary has $O(K^2)$ pixel-corners within distance
$2K+1$.  By the previous lemma all pixel-corners must be within such distance,
so there are $O(K^2n)$ pixel-corners, and hence $O(K^2n)$ pixels.
\end{proof}

\begin{lemma}
\label{lem:Kthin}
\label{lem:kThinRectangles}
Any pixel $\psi$ in a $K$-thin polygon $P$ is intersected by $O(K^2)$ maximal axis-aligned
rectangles inside $P$.
\end{lemma}
\begin{proof}
It suffices to show instead that any pixel-corner $p$ is intersected by $O(K^2)$
such rectangles, because a rectangle intersects
$\psi$ if and only if it intersects one of the four corners of $\psi$.
For ease of description, assume that pixel-corner $p$ is at the origin.
Let $\gs{p}$ to be the region of all points that can $r$-guard $p$.  It
is well-known (see e.g.~\cite{worman2007}) that $\gs{p}$ is an $r$-star,
i.e., it exists of four orthogonal $xy$-monotone chain, one in each of the
quadrants.  (The first and last edge of each chain may lie on a coordinate axis,
and there may be a degenerate ``spike'' along the coordinate axes if degenerate rectangles
are allowed for $r$-guarding.)  
Also, any edge $e$ of $\gs{p}$ lies along pixel-sides, which means that the supporting
line of $e$ is a path of pixelation-edges until the point where it hits the boundary of $P$.

The chain $C_1$ in the first quadrant is monotonically
decreasing in $y$.  We claim that $C_1$
cannot have too many edges.  Assume for contradiction that it had
$4K+4$ or more edges, not counting any edges that are on the $y$-axis
or the $x$-axis.  Enumerate the edges by increasing $x$-coordinate,
and consider the point $q$ common to edge $2K{+2}$ and $2K{+}3$.
By definition of $\gs{p}$ rectangle $R(p,q)$ is inside $P$.
Also, it intersects the supporting lines of vertical edges 
$2,4, \dots,2K{+}2$ (as well as the $y$-axis),
and the supporting lines of horizontal edges
$2K{+}3,2K{+5},\dots,4K+3$ (as well as the $x$-axis).
See Figure~\ref{fig:kThin} for an illustration for $K=2$.
This creates a $(K+2)\times (K+2)$-grid in the pixelation, contradicting that $P$
is $K$-thin.  Likewise we can show that any of the chains $C_2,C_3,C_4$
in the other three quadrants has $O(K)$ edges.

\begin{figure}[t]
\centering
\includegraphics[width=0.55\textwidth]{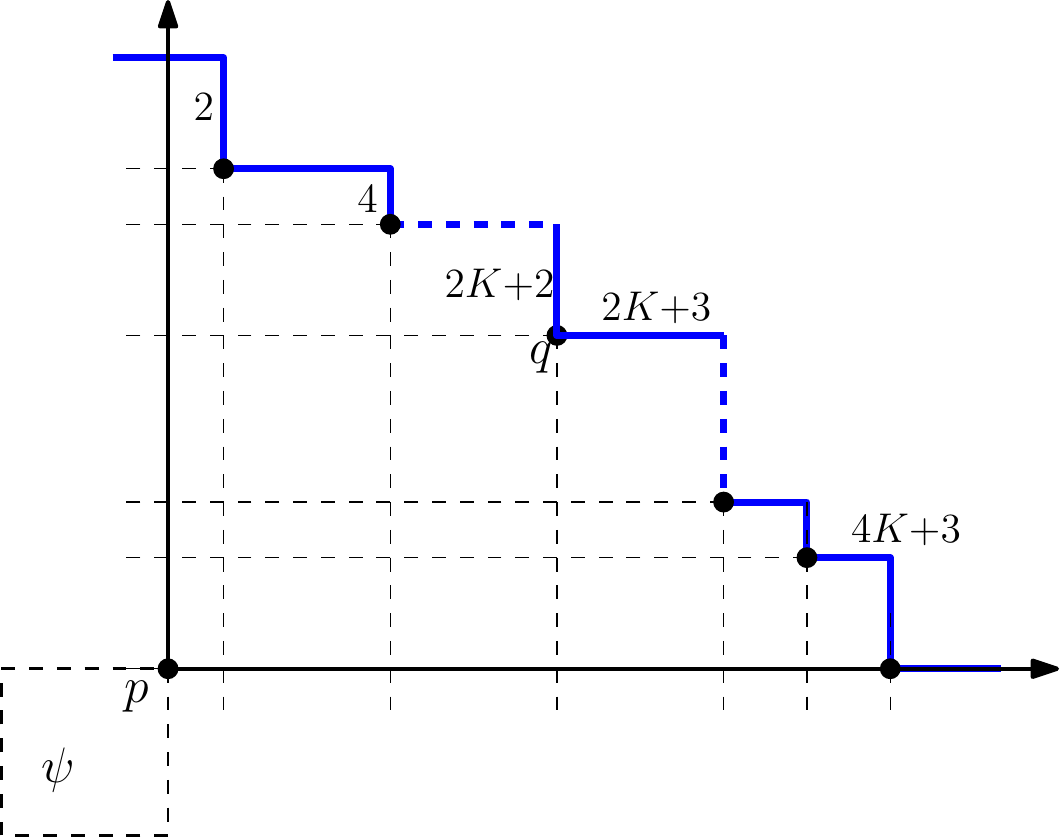}
\caption{If $C_1$ has $4K+4$ edges, then we an find a $(K+2)\times (K+2)$-grid
in the first quadrant.}
\label{fig:kThin}
\end{figure}

Let $\rho$ be a maximal axis-aligned rectangle that intersects $p$.
Clearly $\rho\subseteq \gs{p}$, since any point in $\rho$ $r$-guards $p$.
Also, $\rho$ must touch the boundary of $\gs{p}$ on all four sides, else
it would not be maximal.  Thus $\rho$ must touch at least one vertical
edge in $C_1$ or $C_4$ (on its right side) and at least one vertical
edge in $C_2$ or $C_3$ (on its left side).  Vice versa, if we fix the
vertical edges touched by a rectangle on the left and right side, then
there can be at most one maximal rectangle realizing this (obtained by
expanding upward and downward until we reach the boundary of $\gs{p}$.)  
Therefore, the number of maximal rectangles is at most 
$|C_1\cup C_4|\cdot |C_2\cup C_3|\in O(K^2)$.
\end{proof}

We now have the following:

\begin{theorem}
Let $P$ be a polygon for which the dual graph has treewidth $k$.
Then for any set $U\subseteq P$ and $\Gamma\subseteq P\cup \Psi$,
the $(U,\Gamma,P)$-rGuarding problem can be solved in
$O(f(k^3) k^2n)$ time for some computable function $f(.)$. 
\end{theorem}
\begin{proof}
Since the dual graph $D$ has treewidth $k$, it contains no
$(k+1)\times (k+1)$-grid, and so $P$ is $k$-thin.
The pixelation of $P$ has $O(k^2 n)$ vertices by Lemma~\ref{lem:kThinPixels}, 
and can be constructed in $O(k^2 n)$ time by constructing the
vertical decomposition and then ray-shooting along the horizontal rays
emitted from reflex vertices.  
Find a tree decomposition $\mathcal{T}$ of $D$ of width $O(k)$ with $O(k^2 n)$ bags;
this can be done in linear time since $D$ is planar \cite{KammerTholey2012}.
Replace each pixel in each bag of $\mathcal{T}$ by points, guards
and rectangles as explained in Section~\ref{sec:treedecompConstruction}.  
Since each pixel belongs to
$O(k^2)$ rectangles, the resulting tree decomposition has width
$O(k^3)$.  Now solve the restricted 2-dominating set problem
using Courcelle's theorem.  The run-time is as desired since
we have $O(k^2 n)$ bags and treewidth $O(k^3)$.
\end{proof}

\subsection{$K$-Thin Polygons with Few Holes} 

Both of the above generalizations can be combined, creating an algorithm
that is fixed-parameter tractable in both the thinness and the number of holes.

\begin{lemma}
Let $P$ be a polygon
that is $K$-thin and that has $h$ holes.  Then the dual graph of $P$
has treewidth $O(K(h+1))$.
\end{lemma}
\begin{proof}
Let $D'$ be the (full) dual graph of the pixelation graph, i.e., 
it is graph $D$ plus a vertex for each hold and for the outerface,
connected to all incident pixels.
We claim that all vertices in $D'$ have distance
$O(K(h+1))$ from the outerface-vertex.
This implies that $D'$ (and hence also $D$) is $O(K(h+1))$-outerplanar 
and so has treewidth $O(K(h+1))$.

To prove the distances, we first connect holes as
follows.  If $H$ is a hole, then let $c$ be a corner of $H$ that
maximizes the sum of the coordinates (breaking ties arbitrarily).
Let $\psi$ be a pixel incident to $c$ and let $c'$ be some other
corner of $\psi$.
By Lemma~\ref{lem:kThinDistance}, there exists a pixel-corner
$p$ on the boundary of $P$ within distance $2K+1$ from $c'$. 
Moreover, the path from $c'$ to $p$ goes only up and right.
Thus $p$ is incident to the outer-face or to a hole $H'$,
where $H'\neq H$ by choice of $c$.  Following this path, we
can hence find a path in $D$ of length $O(K)$ from the vertex
representing $H$ to the vertex representing $H'$ or the outer-face.
Combining all these paths, we can reach the outer-face from 
any hole in a path of length $O(K(h+1))$.

Now for any other vertex in $D$ (hence pixel $\psi$), let 
$c$ be one pixel-corner, and find a path in the pixelation of length 
at most $2K+1$ from $c$ to some point on the boundary.  Following
this path, we can find a path of length $O(K)$ in $D$ from $\psi$
to some hole or the outer-face, and hence reach the outer-face
along a path of length $O(K(h+1))$.  The result follows.
\end{proof}

The following summarizes this approach, and includes all previous results.

\begin{theorem}
Let $P$ be a polygon that is $K$-thin and has $h$ holes.
Then for any set $U\subseteq P$ and $\Gamma\subseteq P\cup \Psi$,
the $(U,\Gamma,P)$-rGuarding problem can be solved in
$O(f((K(h+1))^3) (K(h+1))^2n)$ time for some computable function $f(.)$. 
In particular, the rGuarding problem is fixed-parameter
tractable in $K+h$.
\end{theorem}

\section{Conclusion}
\label{sec:conclusion}
In this paper, we studied the problem of guarding a thin polygon
under the model that a guard can only see a point if the entire
axis-aligned rectangle spanned by them is inside the polygon.  We
showed that this problem is NP-hard, even in thin polygons, if there
are holes.  If there are few holes or, more generally, the dual
graph of the polygon has bounded treewidth, then we solved the problem in linear time.

Our approach is quite flexible in that we can specify which points
must be guarded and which points/pixels are allowed to be used as
guards.   In fact, with minor modifications even more flexibility
is possible.  We could allow any guard that consists of a connected
union of pixels (as long as any pixel is intersected by $O(1)$ guards).
We could even consider other guarding models by replacing the
rectangles in $\mathcal{R}$ by arbitrary connected unions of pixels and
pixel-sides (again as
long as any pixel is intersected by $O(1)$ such shapes).  For all these,
the (naturally defined) auxiliary graph $H$ has treewidth $O(tw(D))$
in thin polygons, and we can hence solve $r$-guarding by solving
the restricted distance-2-dominating set.

Our results mean that the complexity of $r$-guarding is
nearly resolved, with the exception of polygons that have $O(1)$
holes but are not $K$-thin for a constant number $K$.
For such polygons, is the problem still NP-hard?
Also, for polygons that have a large number of holes, is the problem
APX-hard, or can we develop a PTAS?

\section*{Acknowledgement}

The authors would like to thank Justin Iwerks and Joe Mitchell;
the discussions with them made us consider thin polygons in the
first place.


\bibliographystyle{plain}
\bibliography{ref}

\begin{appendix}
\newpage

\section{Reducing the number of points and guards}
\label{apx:preliminA}

In this section, we show that we can reduce the number of points and guards to a finite
set without affecting the solutions.  For this, it is important to analyze which points
are equivalent with respect to what they guard.
It was already mentioned by Worman and Keil (\cite{worman2007}, Lemma 1) that if $q$
$r$-guards a point $p$ in a pixel $\psi$, then $q$ guards all of $\psi$.  This, however, is
not correct if $p$ is on the boundary of $\psi$, because $q$ might only see that boundary.
We prove here a more precise version.    As before, let the {\em guarding set} $\gs{p}$
of a point $p$ to be all those points $g$ that $r$-guard $p$.  Since $r$-guarding is
a symmetric operation, this is the same as all those points that are being $r$-guarded
by $p$.

\begin{lemma}
\label{lem:interior2boundary}
Let $p$ be a point in the interior of a pixel $\psi$, and let $p'$ be any other point
in $\psi$ (possibly on the boundary of $\psi$).  Then $\gs{p}\subseteq \gs{p'}$.
\end{lemma}
\begin{proof}
Let $g\in \gs{p}$ be any point that $r$-guards $p$, so $R:=R(g,p) \subseteq P$.  
Define $R'$ to be the
union of all pixels for which an interior point is in $R$.  Clearly $R'$ is again a
rectangle and inside $P$.  Since $p$ is in the interior of $\psi$, $R'$ contains
all of $\psi$, and in particular is
non-degenerate and contains $p$ and $p'$.  Also, line segment $gp$ cannot run along
pixel-boundaries (since $p$ is in the interior) and so must be in the interior
of a pixel in the vicinity of $g$.  Hence $R'$ contains a pixel that contains $g$,
and $g\in R'$.  So $p',g \in R'$ and $R'$ is non-degenerate,
therefore $g$ $r$-guards $p$ and $\gs{p'}$ contains $g$.
\end{proof}

\begin{lemma}
\label{lem:boundary2corner}
Let $p$ be a point in the interior of a side $e$ of a pixel $\psi$, and let $p'$ be any other point
on $e$ (possibly at the end of $e$).  Then $\gs{p}\subseteq \gs{p'}$.
\end{lemma}
\begin{proof}
Let $g\in \gs{p}$ be any point that $r$-guards $p$, so $R:=R(g,p) \subseteq P$.  
If $R$ contains interior points of some pixels near $p$ and $g$, then as above we can expand it
into a rectangle containing $g$ and the entire pixel-side.  So assume that $R$ contains
no interior points of pixels near $p$ or $g$, which means that it is a (horizontal or vertical)
line segment from $g$ to $p$.  Since $p$ is not an endpoint of $e$, this line
segment runs along $e$, and we can expand it to include all of $e$ and in particular $p'$.
The resulting line segment $R'$ satisfies $R'\subseteq P$, and $R'$ is
degenerate if and only if $R$ was.  Thus $g$ $r$-guards $p'$.
\end{proof}

With this, the proof of Lemma~\ref{lem:simplify} is complete.
Now we prove Lemma~\ref{lem:simplifyGuards} that states that we can reduce
the number of point-guards.
\begin{proof}
We construct the set $\Gamma'$ follows. 
\begin{itemize}
\item For every pixel-corner $c$, if $c\in \Gamma$ then add $c$ to $\Gamma'$.
\item For every pixel-side $e$, if neither endpoint of $e$ is in $\Gamma$, but some interior point of $e$ is in $\Gamma$,
	then add one such point to $\Gamma'$,
\item Finally, for every pixel $\psi$, if no corner and no side has points in $\Gamma$, but the interior of $\psi$
	contains points in $\Gamma$, then add one such point to $\Gamma'$.
\end{itemize}
Clearly, $\Gamma'\subseteq \Gamma$, and $\Gamma'$ contains at most 4 points from each pixel. 
Let $S\subseteq \Gamma$ be a set of point-guards, and define $S'\subseteq \Gamma'$ as follows:
If $g\in S$ is a pixel-corner, then add it to $S'$.  If it is in the interior of a pixel-side $e$,
then some point $g'$ on $e$ (possibly at an endpoint) was added to $\Gamma'$;  add that point to $S'$.
By Lemma~\ref{lem:boundary2corner} we have $\gs{g}\subseteq \gs{g'}$, so point $g'$ guards everything that $g$ did.
Finally if $g\in S$ is an interior point of a pixel $\psi$, then some point $g'$ of $\psi$ (possibly
on the boundary) was added to $\Gamma'$; add that point to $S'$.  
By Lemma~\ref{lem:interior2boundary} point $g'$ guards everything that $g$ did.
So $S'$ guards at least as much as $S$ and $|S'|\leq |S|$, which proves the lemma.
\end{proof}

\end{appendix}

\end{document}